\documentclass{article}
\usepackage{amssymb,nicefrac}
\usepackage{amsfonts}
\usepackage{graphicx}
\usepackage[fleqn]{amsmath}
\usepackage[varg]{txfonts}
\usepackage{stmaryrd}

\usepackage{framed}
\usepackage{color}
\usepackage{ulem}
\usepackage{tikzfig}

\tikzstyle{env}=[copoint,regular polygon rotate=0,minimum width=0.2cm, fill=black]

\tikzstyle{every picture}=[baseline=-0.25em]
\tikzstyle{dotpic}=[scale=0.5]
\tikzstyle{diredges}=[every to/.style={diredge}]
\tikzstyle{dot graph}=[shorten <=-0.1mm,shorten >=-0.1mm,scale=0.6]
\tikzstyle{plot point}=[circle,fill=black,minimum width=2mm,inner sep=0]

\tikzstyle{braceedge}=[decorate,decoration={brace,amplitude=2mm,raise=-1mm}]
\tikzstyle{small braceedge}=[decorate,decoration={brace,amplitude=1mm,raise=-1mm}]
\tikzstyle{left hook arrow}=[left hook-latex]
\tikzstyle{right hook arrow}=[right hook-latex]

\tikzstyle{black dot}=[inner sep=0.7mm,minimum width=0pt,minimum height=0pt,fill=black,draw=black,shape=circle]

\tikzstyle{dot}=[black dot]
\tikzstyle{smalldot}=[inner sep=0.4mm,minimum width=0pt,minimum height=0pt,fill=black,draw=black,shape=circle]
\tikzstyle{white dot}=[dot,fill=white]
\tikzstyle{antipode}=[white dot,inner sep=0.3mm,font=\footnotesize]
\tikzstyle{smallwhitedot}=[smalldot,fill=white]
\tikzstyle{alt white dot}=[white dot,label={[xshift=3.07mm,yshift=-0.05mm,font=\footnotesize]left:$*$}]
\tikzstyle{gray dot}=[dot,fill=gray!40!white]
\tikzstyle{smallgraydot}=[smalldot,fill=gray!40!white]
\tikzstyle{box vertex}=[draw=black,rectangle]
\tikzstyle{small box}=[box vertex,fill=white]
\tikzstyle{whitebg}=[fill=white,inner sep=2pt]
\tikzstyle{graph state vertex}=[sg vertex,fill=black]

\tikzstyle{wide copoint}=[fill=white,draw=black,shape=isosceles triangle,shape border rotate=90,isosceles triangle stretches=true,inner sep=1pt,minimum width=1.5cm,minimum height=5mm]
\tikzstyle{wide point}=[fill=white,draw=black,shape=isosceles triangle,shape border rotate=-90,isosceles triangle stretches=true,inner sep=1pt,minimum width=1.5cm,minimum height=4mm]
\tikzstyle{very wide copoint}=[fill=white,draw=black,shape=isosceles triangle,shape border rotate=-90,isosceles triangle stretches=true,inner sep=1pt,minimum width=2.5cm,minimum height=4mm]
\tikzstyle{very wide empty copoint}=[draw=black,shape=isosceles triangle,shape border rotate=-90,isosceles triangle stretches=true,inner sep=1pt,minimum width=2.5cm,minimum height=4mm]
\tikzstyle{symm}=[ultra thick,shorten <=-1mm,shorten >=-1mm]

\tikzstyle{square box}=[rectangle,fill=white,draw=black,minimum height=5mm,minimum width=5mm,font=\small]
\tikzstyle{square gray box}=[rectangle,fill=gray!30,draw=black,minimum height=6mm,minimum width=6mm]
\tikzstyle{copoint}=[regular polygon,regular polygon sides=3,draw=black,scale=0.75,inner sep=-0.5pt,minimum width=7mm,fill=white]
\tikzstyle{point}=[regular polygon,regular polygon sides=3,draw=black,scale=0.75,inner sep=-0.5pt,minimum width=7mm,fill=white,regular polygon rotate=180]
\tikzstyle{gray point}=[point,fill=gray!40!white]
\tikzstyle{gray copoint}=[copoint,fill=gray!40!white]

\newcommand{\edgearrow}{{\arrow[black]{>}}}
\newcommand{\edgetick}{{\arrow[black,scale=0.7,very thick]{|}}}

\tikzstyle{diredge}=[->]
\tikzstyle{rdiredge}=[<-]
\tikzstyle{medium diredge}=[->]

\tikzstyle{short diredge}=[->]
\tikzstyle{halfedge}=[-)]
\tikzstyle{other halfedge}=[(-]
\tikzstyle{freeedge}=[(-)]
\tikzstyle{white edge}=[line width=5pt,white]
\tikzstyle{tick}=[postaction=decorate,decoration={markings, mark=at position 0.5 with \edgetick}]
\tikzstyle{small map edge}=[|-latex, gray!60!blue, shorten <=0.9mm, shorten >=0.5mm]
\tikzstyle{thick dashed edge}=[very thick,dashed,gray!40]
\tikzstyle{map edge}=[|-latex,very thick, gray!40, shorten <=1mm, shorten >=0.5mm]
\tikzstyle{tickedge}=[postaction=decorate,
  decoration={markings, mark=at position 0.5 with \edgetick}]
\tikzstyle{dirtickedge}=[postaction=decorate,
  decoration={markings, mark=at position 0.5 with \edgetick},
  decoration={markings, mark=at position 0.85 with \edgearrow}]
\tikzstyle{dirdoubletickedge}=[postaction=decorate,
  decoration={markings, mark=at position 0.4 with \edgetick},
  decoration={markings, mark=at position 0.6 with \edgetick},
  decoration={markings, mark=at position 0.85 with \edgearrow}]

\makeatletter
\newcommand{\boxshape}[3]{%
\pgfdeclareshape{#1}{
\inheritsavedanchors[from=rectangle] 
\inheritanchorborder[from=rectangle]
\inheritanchor[from=rectangle]{center}
\inheritanchor[from=rectangle]{north}
\inheritanchor[from=rectangle]{south}
\inheritanchor[from=rectangle]{west}
\inheritanchor[from=rectangle]{east}
\backgroundpath{
\southwest \pgf@xa=\pgf@x \pgf@ya=\pgf@y
\northeast \pgf@xb=\pgf@x \pgf@yb=\pgf@y

\@tempdima=#2
\@tempdimb=#3

\pgfpathmoveto{\pgfpoint{\pgf@xa - 5pt + \@tempdima}{\pgf@ya}}
\pgfpathlineto{\pgfpoint{\pgf@xa - 5pt - \@tempdima}{\pgf@yb}}
\pgfpathlineto{\pgfpoint{\pgf@xb + 5pt + \@tempdimb}{\pgf@yb}}
\pgfpathlineto{\pgfpoint{\pgf@xb + 5pt - \@tempdimb}{\pgf@ya}}
\pgfpathlineto{\pgfpoint{\pgf@xa - 5pt + \@tempdima}{\pgf@ya}}
\pgfpathclose
}
}}

\boxshape{NEbox}{0pt}{8pt}
\boxshape{SEbox}{0pt}{-8pt}
\boxshape{NWbox}{8pt}{0pt}
\boxshape{SWbox}{-8pt}{0pt}
\makeatother

\tikzstyle{map}=[draw,shape=NEbox,inner sep=7pt]
\tikzstyle{mapdag}=[draw,shape=SEbox,inner sep=7pt]
\tikzstyle{maptrans}=[draw,shape=SWbox,inner sep=7pt]
\tikzstyle{mapconj}=[draw,shape=NWbox,inner sep=7pt]

\tikzstyle{probs}=[shape=semicircle,fill=gray!40!white,draw=black,shape border rotate=180,minimum width=1.2cm]

\tikzstyle{arrs}=[-latex,font=\small,auto]
\tikzstyle{arrow plain}=[arrs]
\tikzstyle{arrow dashed}=[dashed,arrs]
\tikzstyle{arrow bold}=[very thick,arrs]
\tikzstyle{arrow hide}=[draw=white!0,-]
\tikzstyle{arrow reverse}=[latex-]
\tikzstyle{cdnode}=[]

\tikzstyle{gn}=[dot,fill=green,minimum width=0.3cm,inner sep=0pt]
\tikzstyle{rn}=[dot,fill=red,inner sep=0pt,minimum width=0.3cm]
\tikzstyle{bn}=[dot,fill=blue,minimum width=0.3cm]

\tikzstyle{rc}=[dot,thick,fill=white,draw = red,minimum width=0.3cm,inner sep=0pt]
\tikzstyle{gc}=[dot,thick,fill=white,draw= green,inner sep=0pt,minimum width=0.3cm]
\tikzstyle{bc}=[dot,thick,fill=white,draw= blue,minimum width=0.3cm]

\tikzstyle{label}=[circle,fill=white,minimum width=0.3cm]

\tikzstyle{H box}=[rectangle,fill=yellow,draw=black,xscale=1,yscale=1,font=\small,inner sep=0.75pt]

\tikzstyle{clocklabel}=[dot,fill=yellow,draw=black,font=\tiny,inner sep=0.75pt]

\tikzstyle{rsn}=[circle split,draw,fill=red,font=\tiny,inner sep=0.75pt]
\tikzstyle{gsn}=[circle split,draw,fill=green,font=\tiny,inner sep=0.75pt]
\tikzstyle{bsn}=[circle split,draw,fill=blue,font=\tiny,inner sep=0.75pt]

\tikzstyle{rsc}=[circle split,thick,draw= red,draw,fill=white,font=\tiny,inner sep=0.75pt]
\tikzstyle{gsc}=[circle split,thick,draw= green,draw,fill=white,font=\tiny,inner sep=0.75pt]
\tikzstyle{bsc}=[circle split,thick,draw= blue,draw,fill=white,font=\tiny,inner sep=0.75pt]

\tikzstyle{cnot}=[fill=white,shape=circle,inner sep=-1.4pt]
\tikzstyle{wire label}=[font=\tiny, auto]

\newcommand{\bra}[1]{\ensuremath{\left\langle #1 \right|}}
\newcommand{\ket}[1]{\ensuremath{\left|  #1 \right\rangle}}

\tikzstyle{cdiag}=[matrix of math nodes, row sep=3em, column sep=3em, text height=1.5ex, text depth=0.25ex,inner sep=0.5em]
\tikzstyle{arrow above}=[transform canvas={yshift=0.5ex}]
\tikzstyle{arrow below}=[transform canvas={yshift=-0.5ex}]

\usepackage{mathtools}
\newtheorem{Th}{Theorem}[section]
\newtheorem{theorem}[Th]{Theorem}
\newtheorem{proposition}[Th]{Proposition}
\newtheorem{lemma}[Th]{Lemma}
\newtheorem{corollary}[Th]{Corollary}
\newtheorem{definition}[Th]{Definition}

\newenvironment{proof}{\textbf{Proof:}}{\hfill$\Box$\newline}

\title{Equivalence of Local Complementation and Euler Decomposition in the Qutrit ZX-calculus}
\author{Xiaoyan Gong$^1$\qquad\qquad Quanlong Wang$^{1, ~2}$\\ School of Mathematics and Systems Science, Beihang University$^1$\\ Universit\'e de Lorraine, R\'egion Lorraine,  LORIA$^{ 2}$}

\begin{document}

\date{}\maketitle

\begin{abstract}
In this paper, we give a modified version of the qutrit ZX-calculus, by which we represent qutrit graph states as diagrams and prove that the qutrit version of local complementation property is true if and only if the qutrit Hadamard gate $H$ has an Euler decomposition into $4\pi/3$-green and
red rotations. This paves the way for studying the completeness of qutrit ZX-calculus for qutrit stabilizer quantum mechanics.

\end{abstract}

\section{Introduction}
The ZX-calculus for qubits  is an intuitive and powerful  graphical language. It is an important branch of category quantum mechanics introduced by Coecke and Duncan \cite{CoeckeDuncan}, which allows us to explicitly formulate quantum mechanics within the overall framework of symmetric monoidal categories. This graphical language characterizes complementarity of quantum observables (typically, Pauli Z and X spin observables). In addition, the ZX-calculus is universal and sound for pure qubit quantum mechanics. Although the overall ZX-calculus for pure state qubit quantum mechanics is incomplete \cite{Vladimir},  it is complete for stabilizer quantum mechanics \cite{Miriam1}.

It is well-known that the theory of quantum information and quantum computation is mainly based on qubits. However, qutrits are also useful for quantum information processing. For instance,   there does not exist a (2, 3) quantum secret threshold scheme for qubits in which each share
is also a qubit, while such a scheme for qutrits does exist \cite{Cleve}.  Moreover, better security can be achieved for quantum cryptography
using higher dimensional quantum system \cite{Bourennane}. As the same reason of introducing ZX-calculus for pure state qubit quantum mechanics, it is natural to consider a qutrit version of ZX-calculus. In fact, the ZX-calculus  for qutrits was established in \cite{BianWang1, BianWang2} and independently introduced as a typical special case of qudit ZX-calculus in \cite{Ranchin}.

In this paper, we treat qutrit ZX-calculus in a modified way. Comparing to the rules in \cite{BianWang2}, we add two new rules (S3) and (H2$^\prime$) while remove the (P2) rule which can be derived. The rules (S1) and (P1) are expanded but (K2) is reduced. In contrast to the rules in \cite{Ranchin}, the rules (S3), (P1) and (H2$^\prime$) are added, (S1) is expanded, but (D) is vanishing since we do not consider scalars in this paper. Also, we need not to make the assumption as adopted in \cite{Ranchin} that any topological deformation of the internal structure does not matter. With our current version of rules in hand, we can prove some useful graphical properties of qutrits such as the Hopf law, coincidence of all dualizers and commutativity of copy and co-copy.

Our main result in this paper is to establish the equivalence of local complementation and Hadamard decomposition in the qutrit ZX-calculus. The local complementation \cite{Bouchet,Kant} is a graph transformation which is a powerful tool for the study of graph states. As a class of especial stabilizer states, graph states \cite{Griffiths} is the specific algorithm resources in one-way quantum computing model, and has broad application in quantum information processing.
In \cite{DuncanPerdrix}, graph states are represented by diagrams of the qubit ZX-calculus, and it is further proved that the local complementation
property (Van den Nest's theorem) is true if and only if the Hadamard gate $H$ has an Euler decomposition into $\pi/2$-green and
red rotations. In this paper, we represent qutrit graph states by diagrams of the qutrit ZX-calculus and prove that the qutrit version of local complementation
property is true if and only if the qutrit Hadamard gate $H$ has an Euler decomposition into $4\pi/3$-green and
red rotations. This paves the way for  studying the completeness of qutrit ZX-calculus for qutrit stabilizer quantum mechanics.

\section{Qutrit ZX-calculus}
In this paper, we will work in a symmetric monoidal category and ignore non-zero scalars. First we list the rules of Qutrit ZX-calculus, where the angles $\alpha, \beta, \eta, \theta\in [0,2\pi)$.  Note that all the diagrams should be read from top to bottom.
\begin{figure}[!h]
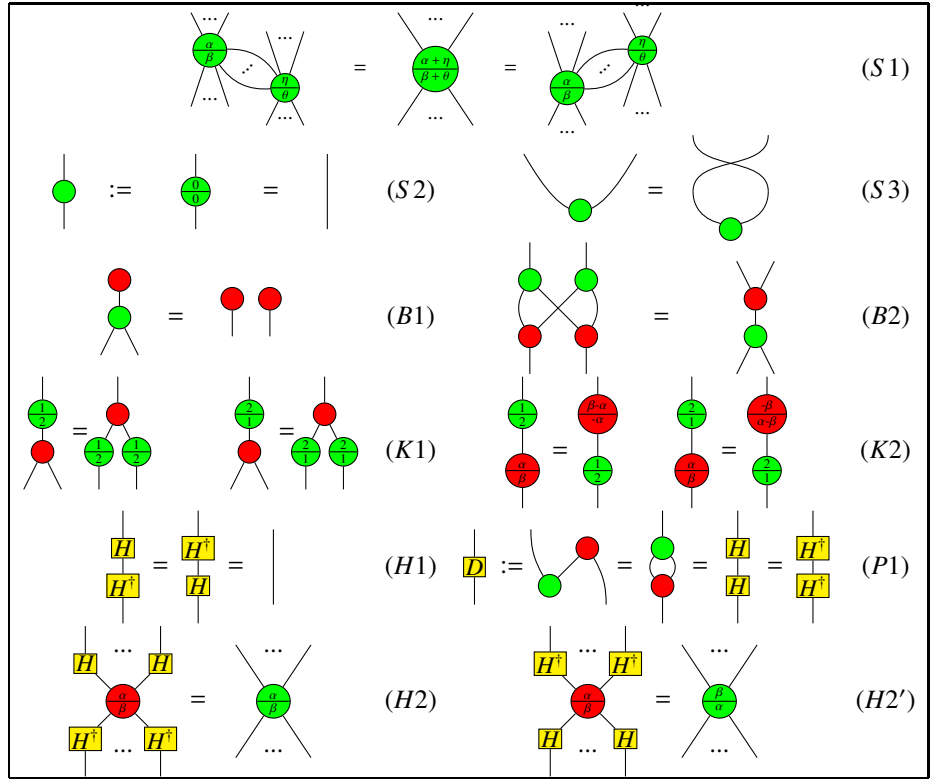

\begin{center}
\[
\quad \qquad\begin{array}{|cccc|}
\hline
\multicolumn{3}{|c}{%
\beginpgfgraphicnamed{Qutrits//spidernew}
\InputIfFileExists{Qutrits//spidernew.tikz}{}{\input{./figures/Qutrits//spidernew.tikz}}
\endpgfgraphicnamed}&(S1)\\
\beginpgfgraphicnamed{RGrelations//s2}
\InputIfFileExists{RGrelations//s2.tikz}{}{\input{./figures/RGrelations//s2.tikz}}
\endpgfgraphicnamed&(S2)&%
\beginpgfgraphicnamed{Qutrits//cupswap}
\InputIfFileExists{Qutrits//cupswap.tikz}{}{\input{./figures/Qutrits//cupswap.tikz}}
\endpgfgraphicnamed&(S3)\\
\beginpgfgraphicnamed{RGrelations//b1}
\InputIfFileExists{RGrelations//b1.tikz}{}{\input{./figures/RGrelations//b1.tikz}}
\endpgfgraphicnamed&(B1)&%
\beginpgfgraphicnamed{RGrelations//b2}
\InputIfFileExists{RGrelations//b2.tikz}{}{\input{./figures/RGrelations//b2.tikz}}
\endpgfgraphicnamed&(B2)\\
\beginpgfgraphicnamed{Qutrits//k1n}
\InputIfFileExists{Qutrits//k1n.tikz}{}{\input{./figures/Qutrits//k1n.tikz}}
\endpgfgraphicnamed&(K1)&%
\beginpgfgraphicnamed{Qutrits//k2n}
\InputIfFileExists{Qutrits//k2n.tikz}{}{\input{./figures/Qutrits//k2n.tikz}}
\endpgfgraphicnamed&(K2)\\
%
\beginpgfgraphicnamed{RGrelations//h1}
\InputIfFileExists{RGrelations//h1.tikz}{}{\input{./figures/RGrelations//h1.tikz}}
\endpgfgraphicnamed&(H1)&%
\beginpgfgraphicnamed{RGrelations//p1s}
\InputIfFileExists{RGrelations//p1s.tikz}{}{\input{./figures/RGrelations//p1s.tikz}}
\endpgfgraphicnamed&(P1)\\
%
\beginpgfgraphicnamed{RGrelations//h2}
\InputIfFileExists{RGrelations//h2.tikz}{}{\input{./figures/RGrelations//h2.tikz}}
\endpgfgraphicnamed&(H2)&%
\beginpgfgraphicnamed{Qutrits//h2prime}
\InputIfFileExists{Qutrits//h2prime.tikz}{}{\input{./figures/Qutrits//h2prime.tikz}}
\endpgfgraphicnamed&(H2^\prime)\\
\hline
\end{array}\]
\end{center}

  \caption{Qutrit ZX-calculus rules}\label{figure1}
\end{figure}

For convenience, we denote the
frequently used angles $\frac{2 \pi}{3}$ and   $\frac{4 \pi}{3}$ by $1$ and $2$ respectively.
The box H is called a Hadamard gate.

The diagrams in Qutrit ZX-calculus have a standard interpretation $\llbracket \cdot \rrbracket$ in the Hilbert spaces:
\begin{center}
\[
\left\llbracket %
\beginpgfgraphicnamed{RGgenerator//RGg_ezd}
\begin{tikzpicture}
	\begin{pgfonlayer}{nodelayer}
		\node [style=none] (0) at (0, -0.5) {};
		\node [style=gn] (1) at (0, -0) {};
	\end{pgfonlayer}
	\begin{pgfonlayer}{edgelayer}
		\draw (0.center) to (1);
	\end{pgfonlayer}
\end{tikzpicture}}
\endpgfgraphicnamed\right\rrbracket=\ket{+} \qquad
\left\llbracket%
\beginpgfgraphicnamed{RGgenerator//RGg_ez}
\begin{tikzpicture}
	\begin{pgfonlayer}{nodelayer}
		\node [style=none] (0) at (0, 0.5) {};
		\node [style=gn] (1) at (0, -0) {};
	\end{pgfonlayer}
	\begin{pgfonlayer}{edgelayer}
		\draw (0.center) to (1);
	\end{pgfonlayer}
\end{tikzpicture}}
\endpgfgraphicnamed\right\rrbracket=\bra{+} \qquad
\left\llbracket%
\beginpgfgraphicnamed{RGgenerator//RGg_dz}
\begin{tikzpicture}
	\begin{pgfonlayer}{nodelayer}
		\node [style=none] (0) at (0, 0.5) {};
		\node [style=none] (1) at (0.25, -0.5) {};
		\node [style=none] (2) at (-0.25, -0.5) {};
		\node [style=gn] (3) at (0, -0) {};
	\end{pgfonlayer}
	\begin{pgfonlayer}{edgelayer}
		\draw (0.center) to (3);
		\draw (3) to (2.center);
		\draw (3) to (1.center);
	\end{pgfonlayer}
\end{tikzpicture}}
\endpgfgraphicnamed\right\rrbracket=\ket{00}\bra{0}+\ket{11}\bra{1}+\ket{22}\bra{2}
\]
\[
\left\llbracket%
\beginpgfgraphicnamed{RGgenerator//RGg_dzd}
\begin{tikzpicture}
	\begin{pgfonlayer}{nodelayer}
		\node [style=gn] (0) at (0, -0) {};
		\node [style=none] (1) at (0, -0.5) {};
		\node [style=none] (2) at (0.25, 0.5) {};
		\node [style=none] (3) at (-0.25, 0.5) {};
	\end{pgfonlayer}
	\begin{pgfonlayer}{edgelayer}
		\draw (0) to (1.center);
		\draw (0) to (2.center);
		\draw (0) to (3.center);
	\end{pgfonlayer}
\end{tikzpicture}}
\endpgfgraphicnamed\right\rrbracket=\ket{0}\bra{00}+\ket{1}\bra{11}+\ket{2}\bra{22} \qquad
\left\llbracket%
\beginpgfgraphicnamed{RGgenerator//RGg_zph_ab}
\begin{tikzpicture}
	\begin{pgfonlayer}{nodelayer}
		\node [style=gsn] (0) at (0, -0) {$\alpha$\nodepart{lower}$\beta$};
		\node [style=none] (1) at (0, 0.5) {};
		\node [style=none] (2) at (0, -0.5) {};
	\end{pgfonlayer}
	\begin{pgfonlayer}{edgelayer}
		\draw (1.center) to (0);
		\draw (2.center) to (0);
	\end{pgfonlayer}
\end{tikzpicture}}
\endpgfgraphicnamed\right\rrbracket=\ket{0}\bra{0}+ e^{i \alpha}\ket{1}\bra{1}
+ e^{i \beta}\ket{2}\bra{2} \qquad
\]

\[
\left\llbracket%
\beginpgfgraphicnamed{RGgenerator//RGg_exd}
\begin{tikzpicture}
	\begin{pgfonlayer}{nodelayer}
		\node [style=none] (0) at (0, -0.5) {};
		\node [style=rn] (1) at (0, -0) {};
	\end{pgfonlayer}
	\begin{pgfonlayer}{edgelayer}
		\draw (0.center) to (1);
	\end{pgfonlayer}
\end{tikzpicture}}
\endpgfgraphicnamed\right\rrbracket=\ket{0} \qquad
\left\llbracket%
\beginpgfgraphicnamed{RGgenerator//RGg_ex}
\begin{tikzpicture}
	\begin{pgfonlayer}{nodelayer}
		\node [style=none] (0) at (0, 0.5) {};
		\node [style=rn] (1) at (0, -0) {};
	\end{pgfonlayer}
	\begin{pgfonlayer}{edgelayer}
		\draw (0.center) to (1);
	\end{pgfonlayer}
\end{tikzpicture}}
\endpgfgraphicnamed\right\rrbracket=\bra{0} \qquad
\left\llbracket%
\beginpgfgraphicnamed{RGgenerator//RGg_dx}
\begin{tikzpicture}
	\begin{pgfonlayer}{nodelayer}
		\node [style=none] (0) at (0, 0.5) {};
		\node [style=none] (1) at (0.25, -0.5) {};
		\node [style=none] (2) at (-0.25, -0.5) {};
		\node [style=rn] (3) at (0, -0) {};
	\end{pgfonlayer}
	\begin{pgfonlayer}{edgelayer}
		\draw (0.center) to (3);
		\draw (3) to (2.center);
		\draw (3) to (1.center);
	\end{pgfonlayer}
\end{tikzpicture}}
\endpgfgraphicnamed\right\rrbracket=\ket{++}\bra{ +}+\ket{\omega\omega}\bra{\omega}+\ket{\bar{\omega}\bar{\omega}}\bra{\bar{\omega}}
\]
\[
\left\llbracket%
\beginpgfgraphicnamed{RGgenerator//RGg_dxd}
\begin{tikzpicture}
	\begin{pgfonlayer}{nodelayer}
		\node [style=rn] (0) at (0, -0) {};
		\node [style=none] (1) at (0, -0.5) {};
		\node [style=none] (2) at (0.25, 0.5) {};
		\node [style=none] (3) at (-0.25, 0.5) {};
	\end{pgfonlayer}
	\begin{pgfonlayer}{edgelayer}
		\draw (0) to (1.center);
		\draw (0) to (2.center);
		\draw (0) to (3.center);
	\end{pgfonlayer}
\end{tikzpicture}}
\endpgfgraphicnamed\right\rrbracket=\ket{+}\bra{++}+\ket{\omega}\bra{\omega\omega}
+\ket{\bar{\omega}}\bra{\bar{\omega}\bar{\omega}} \qquad
\left\llbracket%
\beginpgfgraphicnamed{RGgenerator//RGg_xph_ab}
\begin{tikzpicture}
	\begin{pgfonlayer}{nodelayer}
		\node [style=rsn] (0) at (0, -0) {$\alpha$\nodepart{lower}$\beta$};
		\node [style=none] (1) at (0, 0.5) {};
		\node [style=none] (2) at (0, -0.5) {};
	\end{pgfonlayer}
	\begin{pgfonlayer}{edgelayer}
		\draw (1.center) to (0);
		\draw (2.center) to (0);
	\end{pgfonlayer}
\end{tikzpicture}}
\endpgfgraphicnamed\right\rrbracket=\ket{+}\bra{+}+ e^{i \alpha}\ket{\omega}\bra{\omega}
+ e^{i \beta}\ket{\bar{\omega}}\bra{\bar{\omega}}
\]

\[
\left\llbracket%
\beginpgfgraphicnamed{RGgenerator//RGg_Hada}
\begin{tikzpicture}
	\begin{pgfonlayer}{nodelayer}
		\node [style={H box}] (0) at (0, -0) {$H$};
		\node [style=none] (1) at (0, 0.5) {};
		\node [style=none] (2) at (0, -0.5) {};
	\end{pgfonlayer}
	\begin{pgfonlayer}{edgelayer}
		\draw (1.center) to (0);
		\draw (2.center) to (0);
	\end{pgfonlayer}
\end{tikzpicture}}
\endpgfgraphicnamed\right\rrbracket=\ket{+}\bra{0}+ \ket{\omega}\bra{1}+\ket{\bar{\omega}}\bra{2}=\ket{0}\bra{+}+ \ket{1}\bra{\bar{\omega}}+\ket{2}\bra{\omega}\]
\[
\left\llbracket%
\beginpgfgraphicnamed{RGgenerator//RGg_Hadad}
\begin{tikzpicture}
	\begin{pgfonlayer}{nodelayer}
		\node [style={H box}] (0) at (0, -0) {$H^\dagger$};
		\node [style=none] (1) at (0, -0.5) {};
		\node [style=none] (2) at (0, 0.5) {};
	\end{pgfonlayer}
	\begin{pgfonlayer}{edgelayer}
		\draw (2.center) to (0);
		\draw (1.center) to (0);
	\end{pgfonlayer}
\end{tikzpicture}}
\endpgfgraphicnamed\right\rrbracket=\ket{0}\bra{+}+ \ket{1}\bra{\omega}+\ket{2}\bra{\bar{\omega}}=\ket{+}\bra{0}+ \ket{\omega}\bra{2}+\ket{\bar{\omega}}\bra{1}
\]
\end{center}

where $\omega=e^{\frac{2}{3}\pi i},\bar{\omega}=e^{\frac{4}{3}\pi i}=\omega^2$, and
 \begin{center}
 $
   \left\{\begin{array}{rcl}
      \ket{+} & = & \ket{0}+\ket{1}+\ket{2}\\
      \ket{\omega} & = &  \ket{0}+\omega \ket{1}+\bar{\omega}\ket{2}\\
      \ket{\bar{\omega}} & = & \ket{0}+\bar{\omega}\ket{1}+\omega\ket{2}
    \end{array}\right.
   $
 \end{center}

Note that $\omega^3=1,~~1+\omega+\bar{\omega}=0$. We also use the following matrix form:
  $$\left\llbracket%
\beginpgfgraphicnamed{RGgenerator//RGg_zph_ab}
\begin{tikzpicture}
	\begin{pgfonlayer}{nodelayer}
		\node [style=gsn] (0) at (0, -0) {$\alpha$\nodepart{lower}$\beta$};
		\node [style=none] (1) at (0, 0.5) {};
		\node [style=none] (2) at (0, -0.5) {};
	\end{pgfonlayer}
	\begin{pgfonlayer}{edgelayer}
		\draw (1.center) to (0);
		\draw (2.center) to (0);
	\end{pgfonlayer}
\end{tikzpicture}}
\endpgfgraphicnamed\right\rrbracket=
  \left(
\begin{array}{ccc}
 1 & 0 & 0 \\
 0 & e^{i\alpha} & 0\\
 0 & 0 & e^{i\beta}
\end{array}
\right),\quad\quad\quad
\left\llbracket%
\beginpgfgraphicnamed{RGgenerator//RGg_xph_ab}
\begin{tikzpicture}
	\begin{pgfonlayer}{nodelayer}
		\node [style=rsn] (0) at (0, -0) {$\alpha$\nodepart{lower}$\beta$};
		\node [style=none] (1) at (0, 0.5) {};
		\node [style=none] (2) at (0, -0.5) {};
	\end{pgfonlayer}
	\begin{pgfonlayer}{edgelayer}
		\draw (1.center) to (0);
		\draw (2.center) to (0);
	\end{pgfonlayer}
\end{tikzpicture}}
\endpgfgraphicnamed\right\rrbracket=
  \left(
\begin{array}{ccc}
1+e^{i\alpha}+e^{i\beta} & 1+\bar{\omega} e^{i\alpha}+\omega e^{i\beta}& 1+\omega e^{i\alpha}+\bar{\omega}e^{i\beta} \\
1+\omega e^{i\alpha}+\bar{\omega}e^{i\beta}& 1+e^{i\alpha}+e^{i\beta} & 1+\bar{\omega}e^{i\alpha}+\omega e^{i\beta}\\
 1+\bar{\omega} e^{i\alpha}+\omega e^{i\beta}&1+\omega e^{i\alpha}+\bar{\omega}e^{i\beta} &1+e^{i\alpha}+e^{i\beta}
\end{array}
\right).
$$

Given the standard interpretation $\llbracket \cdot \rrbracket$, it is easy to check that the qutrit ZX-calculus is sound, i.e., any equality that can be derived graphically can also be derived using matrices.

\vspace{0.5cm}
Below we give some properties of the qutrit ZX-calculus.

\begin{lemma}\label{p2rul}
The rule denoted as (P2) in \cite{BianWang2} can be derived:

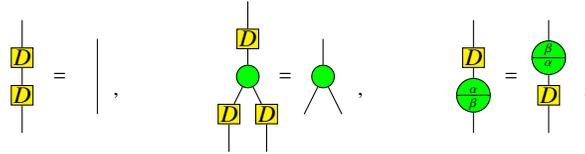
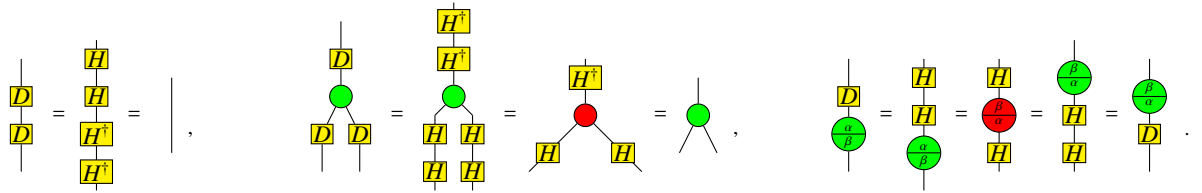
\ctikzfig{Qutrits//p2s}

\end{lemma}
\begin{proof}
\ctikzfig{Qutrits//p2ruleprf}
 Here we used rules (H1), (H2), (P1) and (H2$^\prime$).
\end{proof}

\begin{lemma}\label{swap}
A loop with same colors at the top and bottom is just a identity:
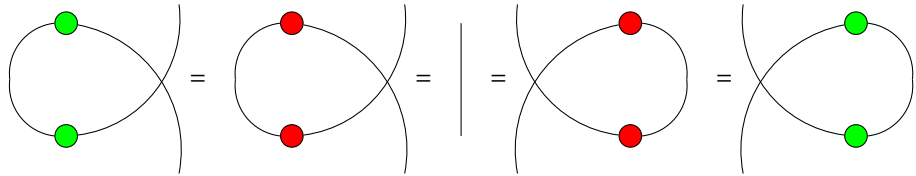
\begin{align*} 
\beginpgfgraphicnamed{Qutrits//swapbendnew}
\InputIfFileExists{Qutrits//swapbendnew.tikz}{}{\input{./figures/Qutrits//swapbendnew.tikz}}
\endpgfgraphicnamed
 \end{align*}
\end{lemma}
\begin{proof}
Since the green parts can be proven in a same way,   we only prove the following equation which follows from (S1) and (S3):
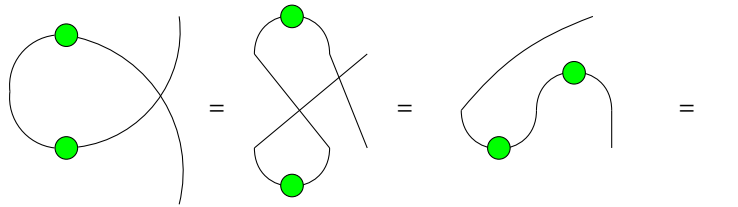
\begin{align*}
\beginpgfgraphicnamed{Qutrits//swapbendprf}
\InputIfFileExists{Qutrits//swapbendprf.tikz}{}{\input{./figures/Qutrits//swapbendprf.tikz}}
\endpgfgraphicnamed
 \end{align*}
The red color parts can be proven from the green parts by the rules (H1), (H2) and (H2$^\prime$).
\end{proof}

\begin{lemma}\label{dualizer}
All the dualizers are coincident:
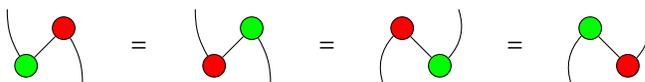
\begin{align*}
\beginpgfgraphicnamed{RGrelations//deri3_4duliser}
\InputIfFileExists{RGrelations//deri3_4duliser.tikz}{}{\input{./figures/RGrelations//deri3_4duliser.tikz}}
\endpgfgraphicnamed
 \end{align*}
\end{lemma}

\begin{proof}
By lemma \ref{swap}, we have
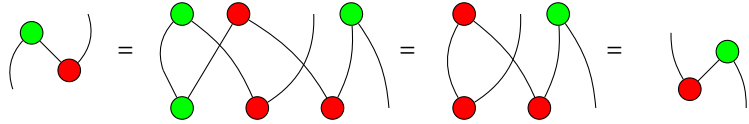
\begin{align*}
\beginpgfgraphicnamed{Qutrits//dualizerprf1}
\InputIfFileExists{Qutrits//dualizerprf1.tikz}{}{\input{./figures/Qutrits//dualizerprf1.tikz}}
\endpgfgraphicnamed
 \end{align*}
In a similar way, we have
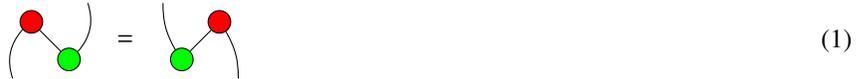
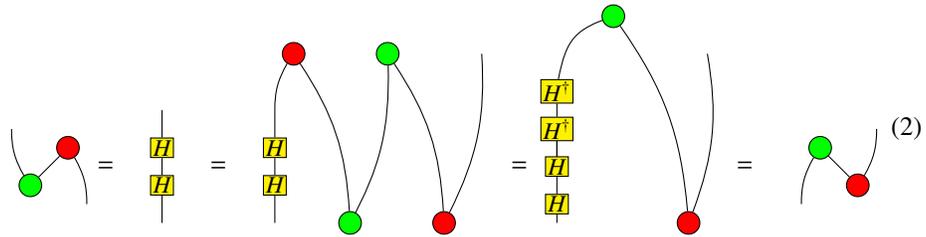
\begin{equation}\label{dualizerpr2}
\beginpgfgraphicnamed{Qutrits//dualizerprf2}
\InputIfFileExists{Qutrits//dualizerprf2.tikz}{}{\input{./figures/Qutrits//dualizerprf2.tikz}}
\endpgfgraphicnamed
 \end{equation}
 Finally,
 \begin{equation}\label{dualizerpr3}
\beginpgfgraphicnamed{Qutrits//dualizerprf3}
\InputIfFileExists{Qutrits//dualizerprf3.tikz}{}{\input{./figures/Qutrits//dualizerprf3.tikz}}
\endpgfgraphicnamed
 \end{equation}
Here we used the rule (P1) and equation (\ref{dualizerpr2}) for the derivation of the third equality in (\ref{dualizerpr3}):
\end{proof}

\begin{lemma}\label{hadslidegn}
The Hadamard gate $H$ as well as its inverse $H^{\dag}$ can slide along a green cup freely:
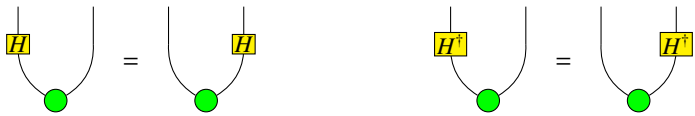
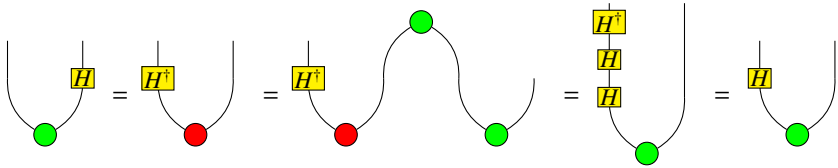
\begin{align*}
\beginpgfgraphicnamed{Qutrits//hadslidegn}
\InputIfFileExists{Qutrits//hadslidegn.tikz}{}{\input{./figures/Qutrits//hadslidegn.tikz}}
\endpgfgraphicnamed
\end{align*}
\end{lemma}

\begin{proof}
\begin{align*}
\beginpgfgraphicnamed{Qutrits//hadslidegnprf}
\InputIfFileExists{Qutrits//hadslidegnprf.tikz}{}{\input{./figures/Qutrits//hadslidegnprf.tikz}}
\endpgfgraphicnamed
\end{align*}
The second equality can be derived similarly using (H2$^\prime$), (H1), (P1) and Lemma \ref{dualizer}.
\end{proof}

Note that the color-swapped versions and upside-down versions of equations in Lemma \ref{hadslidegn} still hold.

In contrast to Lemma \ref{swap}, we have

\begin{lemma}\label{swaprgcapcup}
A loop with different colors at the top and bottom is  a square of the Hadamard gate:
\begin{align*}
\beginpgfgraphicnamed{Qutrits//swaprgcapcup}
\InputIfFileExists{Qutrits//swaprgcapcup.tikz}{}{\input{./figures/Qutrits//swaprgcapcup.tikz}}
\endpgfgraphicnamed
 \end{align*}
\end{lemma}
\begin{proof}
 We only prove the following:
\begin{align*}
\beginpgfgraphicnamed{Qutrits//swaprgcapcupprf}
\InputIfFileExists{Qutrits//swaprgcapcupprf.tikz}{}{\input{./figures/Qutrits//swaprgcapcupprf.tikz}}
\endpgfgraphicnamed
 \end{align*}
 where we used (H2), Lemma \ref{swap} and Lemma \ref{hadslidegn}.
\end{proof}


\begin{lemma}\label{rotationl}
A green dot or phase gate can slide to the opposite side of a red cup with angles upside-down flipped: 
 \begin{equation}\label{rotation}
\beginpgfgraphicnamed{RGrelations//deri2_rotate}
\InputIfFileExists{RGrelations//deri2_rotate.tikz}{}{\input{./figures/RGrelations//deri2_rotate.tikz}}
\endpgfgraphicnamed
 \end{equation}
\end{lemma}

\begin{proof}
\begin{align*}
\beginpgfgraphicnamed{Qutrits//rotationprf1}
\InputIfFileExists{Qutrits//rotationprf1.tikz}{}{\input{./figures/Qutrits//rotationprf1.tikz}}
\endpgfgraphicnamed
 \end{align*}
 \begin{align*}
\beginpgfgraphicnamed{Qutrits//rotationprf2}
\InputIfFileExists{Qutrits//rotationprf2.tikz}{}{\input{./figures/Qutrits//rotationprf2.tikz}}
\endpgfgraphicnamed
 \end{align*}
\end{proof}

Note that the upside-down versions and color-swapped versions of equations of (\ref{rotation}) still hold.

\begin{lemma}\label{copyvars}
The copy rule (B1) has the following variants:
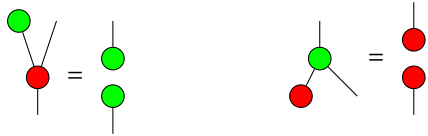
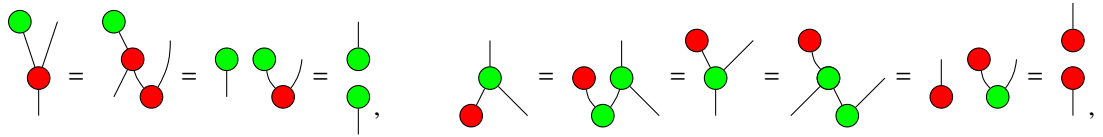
\begin{align*}
\beginpgfgraphicnamed{Qutrits//copyvars}
\InputIfFileExists{Qutrits//copyvars.tikz}{}{\input{./figures/Qutrits//copyvars.tikz}}
\endpgfgraphicnamed
 \end{align*}
\end{lemma}
\begin{proof}
\ctikzfig{Qutrits//copyvarsprf}

 where we used (B1) and its color-swapped version, Lemma \ref{rotationl} as well as  its color-swapped version.
\end{proof}

Now we can have the qutrit version of Hopf law:

\begin{lemma}[Hopf law]
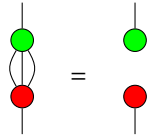
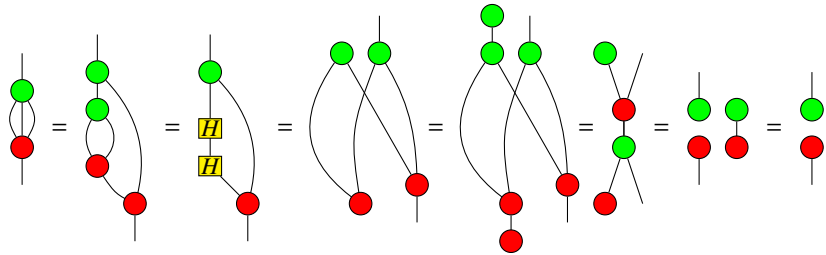
\label{hopf}
\begin{align*}
\beginpgfgraphicnamed{RGrelations//deri1_3break}
\InputIfFileExists{RGrelations//deri1_3break.tikz}{}{\input{./figures/RGrelations//deri1_3break.tikz}}
\endpgfgraphicnamed
 \end{align*}
\end{lemma}

\begin{proof}
\begin{align*}
\beginpgfgraphicnamed{Qutrits//hopfprfnew}
\InputIfFileExists{Qutrits//hopfprfnew.tikz}{}{\input{./figures/Qutrits//hopfprfnew.tikz}}
\endpgfgraphicnamed
 \end{align*}
 Here we used (P1), (B1), (B2), Lemma \ref{copyvars} and Lemma \ref{swaprgcapcup}.
\end{proof}
Note that the color-swapped version of the equation in Lemma \ref{hopf} still hold.

\begin{lemma}\label{p1flip}
A green copy connected with a red co-copy is the same as its upside-down version:
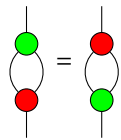
\ctikzfig{Qutrits//p1flip}

\end{lemma}
\begin{proof}
It follows directly from the rule (P1), (H1) and the color change rules (H2) and ($H2^\prime$).
\end{proof}

Similar to the qubit case proof of commutativity of a green copy  in  \cite{bpw}, we can also prove that a green copy in qutrit ZX-calculus is commutative.

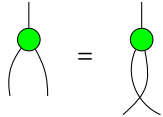
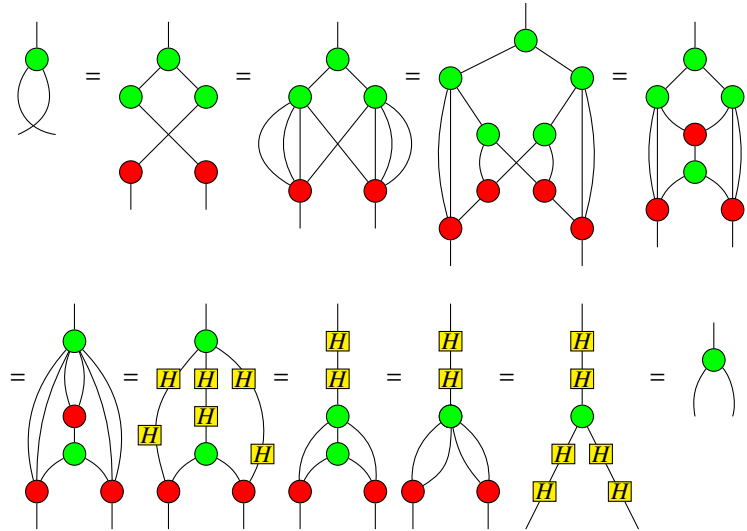
\begin{lemma}\label{greencopy}
\begin{align*}
\beginpgfgraphicnamed{Qutrits//greencommutenew}
\InputIfFileExists{Qutrits//greencommutenew.tikz}{}{\input{./figures/Qutrits//greencommutenew.tikz}}
\endpgfgraphicnamed
 \end{align*}
\end{lemma}

\begin{proof}
\begin{align*}
\beginpgfgraphicnamed{Qutrits//cogreencommuteprf}
\InputIfFileExists{Qutrits//cogreencommuteprf.tikz}{}{\input{./figures/Qutrits//cogreencommuteprf.tikz}}
\endpgfgraphicnamed
 \end{align*}
Here we used (S2), (B2), Lemma \ref{hopf}, (P1), Lemma \ref{p2rul} and (H1).
\end{proof}

Note that and its color-swapped and upside-down versions still hold.

\begin{lemma}\label{controlnotslide}
There is no horizontally connected controlled-NOT gate in the qutrit ZX-calculus:
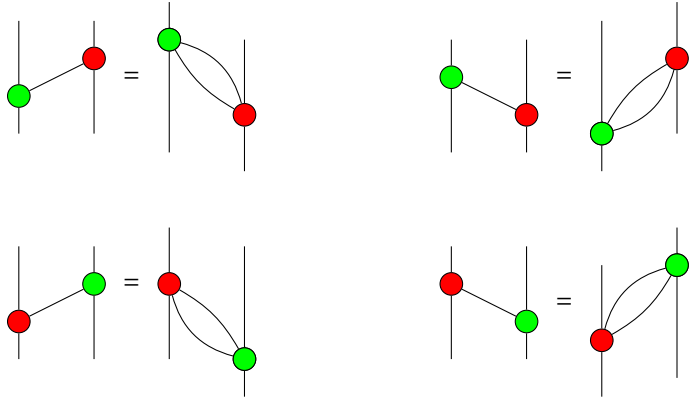
\begin{align*}
\beginpgfgraphicnamed{Qutrits//controlnotslide}
\InputIfFileExists{Qutrits//controlnotslide.tikz}{}{\input{./figures/Qutrits//controlnotslide.tikz}}
\endpgfgraphicnamed
\end{align*}
\end{lemma}

\begin{proof}
We only prove the first equality, the others can be proven similarly.
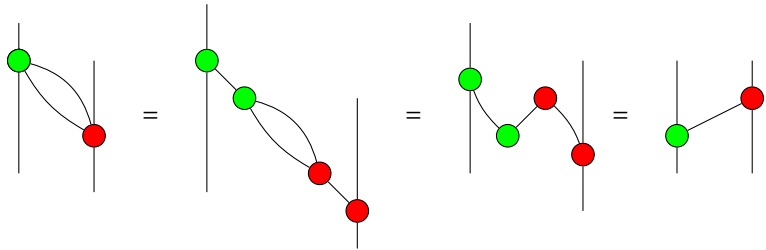
\begin{align*}
\beginpgfgraphicnamed{Qutrits//controlnotslideprf}
\InputIfFileExists{Qutrits//controlnotslideprf.tikz}{}{\input{./figures/Qutrits//controlnotslideprf.tikz}}
\endpgfgraphicnamed
 \end{align*}

\end{proof}

\begin{lemma}\label{quasibialgebra}
There is another form of the bialgebra rule:
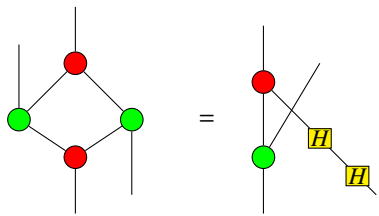
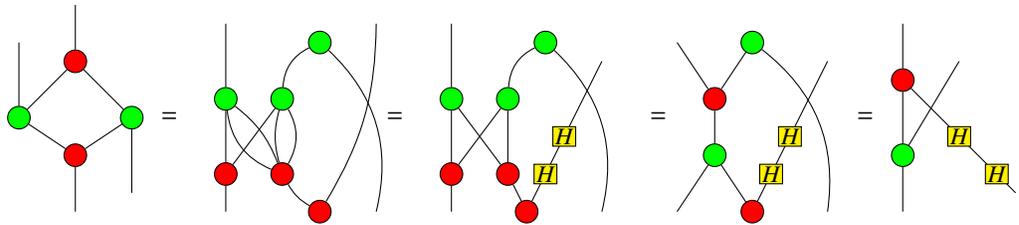
\begin{align*}
\beginpgfgraphicnamed{Qutrits//quasibialgebra}
\InputIfFileExists{Qutrits//quasibialgebra.tikz}{}{\input{./figures/Qutrits//quasibialgebra.tikz}}
\endpgfgraphicnamed
\end{align*}
\end{lemma}
\begin{proof}
\begin{align*}
\beginpgfgraphicnamed{Qutrits//quasibialgebraprf}
\InputIfFileExists{Qutrits//quasibialgebraprf.tikz}{}{\input{./figures/Qutrits//quasibialgebraprf.tikz}}
\endpgfgraphicnamed
 \end{align*}
Here we used (B2), Lemma \ref{p2rule}, Lemma \ref{p1flip}, Lemma \ref{hadslidegn} and Lemma \ref{controlnotslide}.
\end{proof}

By the color change rules, we get the color-swapped version of the equation in this lemma.

\section{Qutrit graph states}

Let $G = (V,E)$ be a graph with $n$ vertices $V$ , each
corresponding to a qutrit, and a collection $E$ of undirected
edges connecting pairs of distinct vertices (no self loops).
Multiple edges are allowed, as long as the multiplicity (weight) does not exceed 2.

\begin{definition}[Graph State]\cite{Griffiths}
A qutrit graph state can be defined as 
$$\ket{G}=\mathcal{U}(\ket{+}^{\otimes n}),$$ 
where $\ket{+}=\frac{1}{\sqrt{3}}(\ket{0}+\ket{1}+\ket{2})$,
$\mathcal{U}=\prod_{\{l,m\}\in E}(C_{lm})^{\Gamma_{lm}}$,  $C_{lm}=\Sigma_{j=0}^{2}\Sigma_{k=0}^{2}\omega^{jk}\ket{jk}\bra{jk}$, $\omega=e^{i\frac{2\pi}{3}}$, the $lm$ element $\Gamma_{lm}$
of the adjacency matrix $\Gamma$ is the number of edges connecting vertex $l$ with vertex $m$.
\end{definition}

 By the standard interpretation, we can check that
\[
C_{lm}=\left\llbracket%
\beginpgfgraphicnamed{Qutrits//controlz}
\InputIfFileExists{Qutrits//controlz.tikz}{}{\input{./figures/Qutrits//controlz.tikz}}
\endpgfgraphicnamed\right\rrbracket=\left\llbracket %
\beginpgfgraphicnamed{Qutrits//controlz2}
\InputIfFileExists{Qutrits//controlz2.tikz}{}{\input{./figures/Qutrits//controlz2.tikz}}
\endpgfgraphicnamed\right\rrbracket
\]

Graphically, we have
\begin{lemma}\label{controlz}
\begin{align*}
\beginpgfgraphicnamed{Qutrits//controlz}
\InputIfFileExists{Qutrits//controlz.tikz}{}{\input{./figures/Qutrits//controlz.tikz}}
\endpgfgraphicnamed=%
\beginpgfgraphicnamed{Qutrits//controlz2}
\InputIfFileExists{Qutrits//controlz2.tikz}{}{\input{./figures/Qutrits//controlz2.tikz}}
\endpgfgraphicnamed=: %
\beginpgfgraphicnamed{Qutrits//controlz3}
\InputIfFileExists{Qutrits//controlz3.tikz}{}{\input{./figures/Qutrits//controlz3.tikz}}
\endpgfgraphicnamed
\end{align*}
\end{lemma}

\begin{proof}
\begin{align*}
\beginpgfgraphicnamed{Qutrits//controlzprf}
\InputIfFileExists{Qutrits//controlzprf.tikz}{}{\input{./figures/Qutrits//controlzprf.tikz}}
\endpgfgraphicnamed
 \end{align*}

\end{proof}

By direct calculation,
\[
C_{lm}^2 =\Sigma_{j=0}^{2}\Sigma_{k=0}^{2}\omega^{2jk}\ket{jk}\bra{jk}=\left\llbracket%
\beginpgfgraphicnamed{Qutrits//controlzsquare}
\InputIfFileExists{Qutrits//controlzsquare.tikz}{}{\input{./figures/Qutrits//controlzsquare.tikz}}
\endpgfgraphicnamed\right\rrbracket=\left\llbracket %
\beginpgfgraphicnamed{Qutrits//controlzsquare2}
\InputIfFileExists{Qutrits//controlzsquare2.tikz}{}{\input{./figures/Qutrits//controlzsquare2.tikz}}
\endpgfgraphicnamed\right\rrbracket.
\]

Similar to the proof of Lemma \ref{controlz}, we have
\begin{lemma}\label{controlzsquare}
\begin{align*}
\beginpgfgraphicnamed{Qutrits//controlzsquare}
\InputIfFileExists{Qutrits//controlzsquare.tikz}{}{\input{./figures/Qutrits//controlzsquare.tikz}}
\endpgfgraphicnamed=%
\beginpgfgraphicnamed{Qutrits//controlzsquare2}
\InputIfFileExists{Qutrits//controlzsquare2.tikz}{}{\input{./figures/Qutrits//controlzsquare2.tikz}}
\endpgfgraphicnamed=: %
\beginpgfgraphicnamed{Qutrits//controlzsquare3}
\InputIfFileExists{Qutrits//controlzsquare3.tikz}{}{\input{./figures/Qutrits//controlzsquare3.tikz}}
\endpgfgraphicnamed
\end{align*}
\end{lemma}

Furthermore, $C_{lm}^3 =\Sigma_{j=0}^{2}\Sigma_{k=0}^{2}\omega^{3jk}\ket{jk}\bra{jk}=I$. Graphically, we have

\begin{lemma}\label{twocontrolz}
\begin{align*}
\beginpgfgraphicnamed{Qutrits//twocontrolz3}
\InputIfFileExists{Qutrits//twocontrolz3.tikz}{}{\input{./figures/Qutrits//twocontrolz3.tikz}}
\endpgfgraphicnamed= %
\beginpgfgraphicnamed{Qutrits//controlzsquare3}
\InputIfFileExists{Qutrits//controlzsquare3.tikz}{}{\input{./figures/Qutrits//controlzsquare3.tikz}}
\endpgfgraphicnamed,\quad\quad %
\beginpgfgraphicnamed{Qutrits//twocontrolzsquare3}
\InputIfFileExists{Qutrits//twocontrolzsquare3.tikz}{}{\input{./figures/Qutrits//twocontrolzsquare3.tikz}}
\endpgfgraphicnamed= %
\beginpgfgraphicnamed{Qutrits//controlz3}
\InputIfFileExists{Qutrits//controlz3.tikz}{}{\input{./figures/Qutrits//controlz3.tikz}}
\endpgfgraphicnamed,
\quad\quad %
\beginpgfgraphicnamed{Qutrits//hhdaggern}
\InputIfFileExists{Qutrits//hhdaggern.tikz}{}{\input{./figures/Qutrits//hhdaggern.tikz}}
\endpgfgraphicnamed= %
\beginpgfgraphicnamed{Qutrits//hdaggerhn}
\InputIfFileExists{Qutrits//hdaggerhn.tikz}{}{\input{./figures/Qutrits//hdaggerhn.tikz}}
\endpgfgraphicnamed= %
\beginpgfgraphicnamed{Qutrits//parallel}
\begin{tikzpicture}
	\begin{pgfonlayer}{nodelayer}
		\node [style=none] (0) at (-0.25, 0.75) {};
		\node [style=none] (1) at (-0.25, -0.75) {};
		\node [style=none] (2) at (0.25, 0.75) {};
		\node [style=none] (3) at (0.25, -0.75) {};
	\end{pgfonlayer}
	\begin{pgfonlayer}{edgelayer}
		\draw (0.center) to (1.center);
		\draw (2.center) to (3.center);
	\end{pgfonlayer}
\end{tikzpicture}}
\endpgfgraphicnamed.
\end{align*}


\end{lemma}
\begin{proof}
We only prove the first equation.
\begin{align*}
\beginpgfgraphicnamed{Qutrits//twocontrolzprf}
\InputIfFileExists{Qutrits//twocontrolzprf.tikz}{}{\input{./figures/Qutrits//twocontrolzprf.tikz}}
\endpgfgraphicnamed
 \end{align*}

\end{proof}

Now we can represent any qutrit graph state in the ZX-calculus:
\begin{proposition}
A qutrit graph state $\ket{G}$, where $G = (E;V)$ is an n-vertex graph, is represented in the graphical
calculus as follows:
\begin{itemize}
 \item  for each vertex $v \in V$, a green node with one output, and
\item  for each single edge $\{u,v\}\in E$, an $H$ node connected to the green nodes representing vertices $u$
and $v$,
\item  for each double edge $\{u,v\}\in E$, an $H^{\dagger}$ node connected to the green nodes representing vertices $u$
and $v$.

\end{itemize}
\end{proposition}

\section{ Local Complementation}
In this section, we prove the equivalence of local complementation
and Hadamard decomposition in the qutrit ZX-calculus.

\begin{definition}[Local Complementation]\cite{mmp}
Let $G = (V,\Gamma)$ be a multiple graph with adjacency matrix $\Gamma$ and multiplicity not exceed 2, $\lambda\in \{1, 2\}$. The $\lambda$-local complementation at the vertex $u$ is the is the multigraph $G*^{\lambda}u=(V,\Gamma^\prime)$ such that $\forall v,w \in V, v \neq w, \Gamma^\prime(v,w) = [\Gamma(v,w) + \lambda\Gamma(v, u)\cdot\Gamma(u,w)] (mod 3)$.
\end{definition}

For example, let $G$ be the following graph
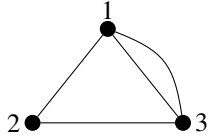
\begin{align*}
\beginpgfgraphicnamed{Qutrits//triangle}
\InputIfFileExists{Qutrits//triangle.tikz}{}{\input{./figures/Qutrits//triangle.tikz}}
\endpgfgraphicnamed
\end{align*}
The corresponding graph state $\ket{G}$ is
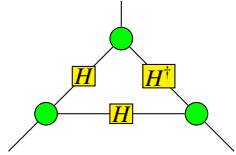
\begin{align*}
\beginpgfgraphicnamed{Qutrits//trianglestate}
\InputIfFileExists{Qutrits//trianglestate.tikz}{}{\input{./figures/Qutrits//trianglestate.tikz}}
\endpgfgraphicnamed
\end{align*}
Its 1-local complementation $G*^{1}1$ at the vertex $1$ is
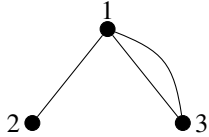
\begin{align*}
\beginpgfgraphicnamed{Qutrits//trianglelocalcom}
\InputIfFileExists{Qutrits//trianglelocalcom.tikz}{}{\input{./figures/Qutrits//trianglelocalcom.tikz}}
\endpgfgraphicnamed
\end{align*}
The corresponding graph state $\ket{G*^{1}1}$ is
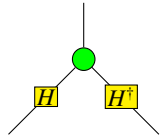
\begin{align*}
\beginpgfgraphicnamed{Qutrits//trianglelocalcomstate}
\InputIfFileExists{Qutrits//trianglelocalcomstate.tikz}{}{\input{./figures/Qutrits//trianglelocalcomstate.tikz}}
\endpgfgraphicnamed
\end{align*}
\begin{theorem}[Local Complementation Property]\cite{Keet}\label{lcproperty}
For any 3-multigraph  $G = (V,\Gamma)$, any $u\in V$, and any $\lambda\in \{1, 2\}$, there exists a local unitary transformation $U$
 such that $\ket{G*^{\lambda}u}=U\ket{G}$.
\end{theorem}

\begin{theorem}\label{lceuler}
The local complementation property is true if and only if $H$ can be decomposed as follows:

\ctikzfig{Qutrits//HadaDecom}
\end{theorem}

 This equation will be called the Euler decomposition of $H$, as similar to the qubit case. We give some consequences of this decomposition  in the following.

\begin{lemma}\label{eulerconsequence}
The $H$-decomposition and the $H^{\dag}$-decomposition are not unique:
\begin{align*}
\beginpgfgraphicnamed{Qutrits//HadaDecom}
\InputIfFileExists{Qutrits//HadaDecom.tikz}{}{\input{./figures/Qutrits//HadaDecom.tikz}}
\endpgfgraphicnamed~~\Rightarrow ~~~~%
\beginpgfgraphicnamed{Qutrits//HadaDecom2}
\InputIfFileExists{Qutrits//HadaDecom2.tikz}{}{\input{./figures/Qutrits//HadaDecom2.tikz}}
\endpgfgraphicnamed ~, ~~~~%
\beginpgfgraphicnamed{Qutrits//hdagerdecom}
\InputIfFileExists{Qutrits//hdagerdecom.tikz}{}{\input{./figures/Qutrits//hdagerdecom.tikz}}
\endpgfgraphicnamed.
\end{align*}
\end{lemma}

\begin{proof}
\begin{align*}
\beginpgfgraphicnamed{Qutrits//eulerconsequenceprf}
\InputIfFileExists{Qutrits//eulerconsequenceprf.tikz}{}{\input{./figures/Qutrits//eulerconsequenceprf.tikz}}
\endpgfgraphicnamed
\end{align*}
\end{proof}

\begin{lemma}\label{rotationchange}
One color of $4\pi/3$ rotation can be expressed in terms of the other color. 
\begin{align*}
\beginpgfgraphicnamed{Qutrits//HadaDecom}
\InputIfFileExists{Qutrits//HadaDecom.tikz}{}{\input{./figures/Qutrits//HadaDecom.tikz}}
\endpgfgraphicnamed~~\Rightarrow ~~~~%
\beginpgfgraphicnamed{Qutrits//rotationchange}
\InputIfFileExists{Qutrits//rotationchange.tikz}{}{\input{./figures/Qutrits//rotationchange.tikz}}
\endpgfgraphicnamed
\end{align*}
\end{lemma}

\begin{proof}
\begin{align*}
\beginpgfgraphicnamed{Qutrits//rotationchangeprf}
\InputIfFileExists{Qutrits//rotationchangeprf.tikz}{}{\input{./figures/Qutrits//rotationchangeprf.tikz}}
\endpgfgraphicnamed
\end{align*}
Here we used Lemma \ref{eulerconsequence} and Lemma \ref{copyvars}.
\end{proof}

The rest of the paper is dedicated to the proof of  Theorem \ref{lceuler}: the equivalence
of local complementation property and the Euler decomposition of H. Note that below we only consider $1$-local complementation, since the case of $2$-local complementation is similar.

\subsection{Local Complementation Implies Euler Decomposition}

\begin{lemma}\label{localimplyhad}
Local complementation of a triangle implies the H-decomposition:
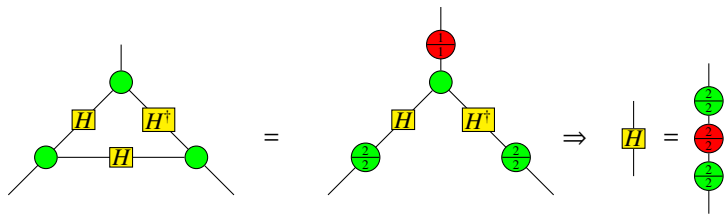
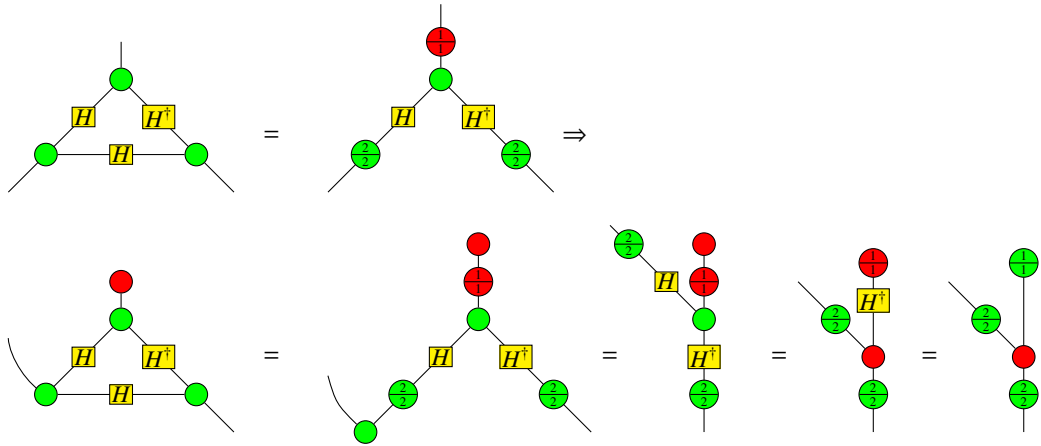
\begin{align*}
\beginpgfgraphicnamed{Qutrits//localcomtri}
\InputIfFileExists{Qutrits//localcomtri.tikz}{}{\input{./figures/Qutrits//localcomtri.tikz}}
\endpgfgraphicnamed \Rightarrow ~~~~ %
\beginpgfgraphicnamed{Qutrits//HadaDecom}
\InputIfFileExists{Qutrits//HadaDecom.tikz}{}{\input{./figures/Qutrits//HadaDecom.tikz}}
\endpgfgraphicnamed
\end{align*}
\end{lemma}

\begin{proof}
\begin{align*}
\beginpgfgraphicnamed{Qutrits//localcomtri}
\InputIfFileExists{Qutrits//localcomtri.tikz}{}{\input{./figures/Qutrits//localcomtri.tikz}}
\endpgfgraphicnamed \Rightarrow \end{align*}
\begin{align*}
\beginpgfgraphicnamed{Qutrits//localcomtriprf1}
\InputIfFileExists{Qutrits//localcomtriprf1.tikz}{}{\input{./figures/Qutrits//localcomtriprf1.tikz}}
\endpgfgraphicnamed
\end{align*}
Here we used lemma \ref{hadslidegn}.
 On the other hand,
 \begin{align*}
\beginpgfgraphicnamed{Qutrits//localcomtriprf2}
\InputIfFileExists{Qutrits//localcomtriprf2.tikz}{}{\input{./figures/Qutrits//localcomtriprf2.tikz}}
\endpgfgraphicnamed
\end{align*}
Therefore,
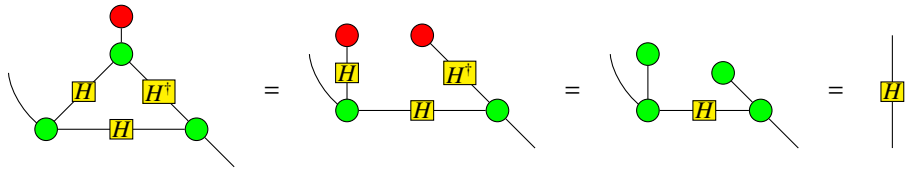
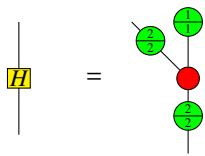
\begin{align*}
\beginpgfgraphicnamed{Qutrits//localcomtriprf3}
\InputIfFileExists{Qutrits//localcomtriprf3.tikz}{}{\input{./figures/Qutrits//localcomtriprf3.tikz}}
\endpgfgraphicnamed
\end{align*}
Applying the color change rules to the above equation, we have
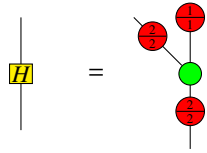
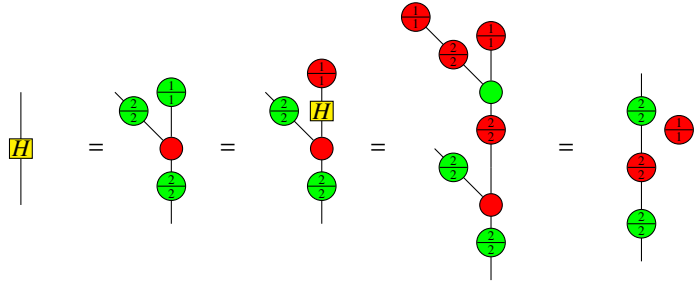
\begin{align*}
\beginpgfgraphicnamed{Qutrits//localcomtriprf4}
\InputIfFileExists{Qutrits//localcomtriprf4.tikz}{}{\input{./figures/Qutrits//localcomtriprf4.tikz}}
\endpgfgraphicnamed
\end{align*}
Finally,
\begin{align*}
\beginpgfgraphicnamed{Qutrits//localcomtriprf5}
\InputIfFileExists{Qutrits//localcomtriprf5.tikz}{}{\input{./figures/Qutrits//localcomtriprf5.tikz}}
\endpgfgraphicnamed
\end{align*}
which is the desired decomposition.
\end{proof}

\subsection{ Euler Decomposition Implies Local Complementation}
We begin with the simplest non trivial examples of local complementation,
namely triangles (with one multiple edge). We need the following
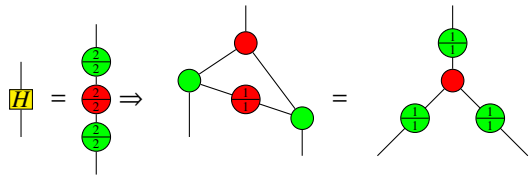
\begin{lemma}\label{trianglebila}
\begin{align*}
\beginpgfgraphicnamed{Qutrits//HadaDecom}
\InputIfFileExists{Qutrits//HadaDecom.tikz}{}{\input{./figures/Qutrits//HadaDecom.tikz}}
\endpgfgraphicnamed  \Rightarrow ~~~~%
\beginpgfgraphicnamed{Qutrits//trianglebila}
\InputIfFileExists{Qutrits//trianglebila.tikz}{}{\input{./figures/Qutrits//trianglebila.tikz}}
\endpgfgraphicnamed
\end{align*}
\end{lemma}

The proof of Lemma \ref{trianglebila} is given in appendix.

\begin{lemma}\label{hadimplylocalcomple}
A local complementation on the top vertex of the triangle
removes the opposite edge.
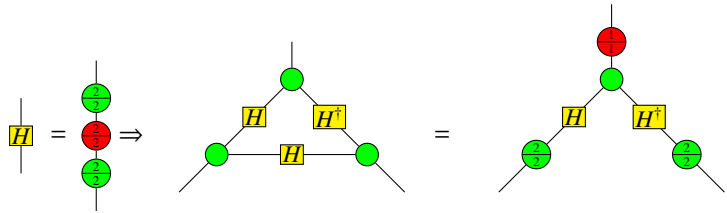
\begin{align*}
\beginpgfgraphicnamed{Qutrits//HadaDecom}
\InputIfFileExists{Qutrits//HadaDecom.tikz}{}{\input{./figures/Qutrits//HadaDecom.tikz}}
\endpgfgraphicnamed  \Rightarrow ~~~~%
\beginpgfgraphicnamed{Qutrits//localcomtri}
\InputIfFileExists{Qutrits//localcomtri.tikz}{}{\input{./figures/Qutrits//localcomtri.tikz}}
\endpgfgraphicnamed
\end{align*}
\end{lemma}
The proof of Lemma \ref{hadimplylocalcomple} is given in appendix.

\begin{lemma}\label{hadimplylocalcomple2}
A local complementation on the top vertex of another form of the triangle
removes the opposite edge.
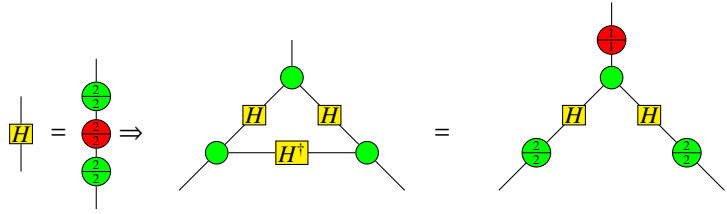
\begin{align*}
\beginpgfgraphicnamed{Qutrits//HadaDecom}
\InputIfFileExists{Qutrits//HadaDecom.tikz}{}{\input{./figures/Qutrits//HadaDecom.tikz}}
\endpgfgraphicnamed  \Rightarrow ~~~~%
\beginpgfgraphicnamed{Qutrits//localhhhd}
\InputIfFileExists{Qutrits//localhhhd.tikz}{}{\input{./figures/Qutrits//localhhhd.tikz}}
\endpgfgraphicnamed
\end{align*}
\end{lemma}

The proof of Lemma \ref{hadimplylocalcomple2} is given in appendix.

\begin{lemma}\label{hadimplylocalcomplestar} {\bf Complete graphs and Star graphs}
Let $K_n$ ($n>1$) be a complete graph with each pair of green nodes connected by $H^{\dag}$ nodes. Then
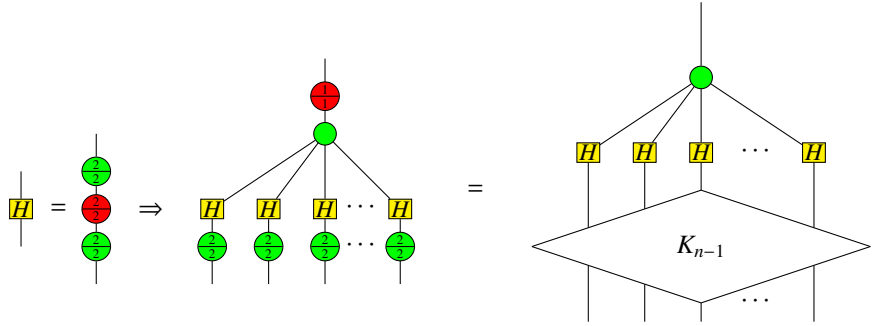
\begin{align*}
\beginpgfgraphicnamed{Qutrits//HadaDecom}
\InputIfFileExists{Qutrits//HadaDecom.tikz}{}{\input{./figures/Qutrits//HadaDecom.tikz}}
\endpgfgraphicnamed ~~~ \Rightarrow ~~~~%
\beginpgfgraphicnamed{Qutrits//stargraph}
\InputIfFileExists{Qutrits//stargraph.tikz}{}{\input{./figures/Qutrits//stargraph.tikz}}
\endpgfgraphicnamed
\end{align*}
\end{lemma}

The proof is quite similar to that of Lemma 8 in \cite{DuncanPerdrix}, we give it in the appendix.

{\bf General case} The general case can be derived from the previous case: we only need to consider the neighbors of the top vertex $u$ (non-neighbor part keep unchanged). According to Lemma \ref{twocontrolz}, if there is no edge connecting two green vertices in the  neighbors of $u$, then it can be rewritten as connected by an $H^{\dag}$-edge and an $H$-edge; if there is an $H$-edge connecting two green vertices, then it can be rewritten as connected by two $H^{\dag}$-edges. Therefore, the  neighbor part of  $u$ can be rewritten as a complete graph with each pair of green nodes connected by $H^{\dag}$ nodes and some more  $H^{\dag}$ or $H$ edges that can be pushed to the bottom of the graph. We can do similar things to the edges connecting $u$ and its neighbors. Thus a general graph can be rewritten as a Star graph with more $H^{\dag}$ or $H$ edges that can be applied after local complementation of the star graph.

For $\lambda=2$, we can prove in a similar way that euler decomposition and local complementation are equivalent.

\section{Conclusion and further work}

In this paper, we modify the rules of qutrit ZX-calculus by adding (S3) and (H2$^\prime$), expading (S1) and (P1), removing (P2) and reducing (K2) in comparing with that of \cite{BianWang2}. With these rules, we prove some useful graphical properties of qutrits such as the Hopf law, coincidence of all dualizers and commutativity of copy and co-copy. As a main result,  we represent qutrit graph states by diagrams of the qutrit ZX-calculus and prove that the qutrit version of local complementation
property is true if and only if the qutrit Hadamard gate $H$ has an Euler decomposition into $4\pi/3$-green and
red rotations.

Given this modified version of the qutrit ZX-calculus, it is natural to consider the completeness of qutrit ZX-calculus for qutrit stabilizer quantum mechanics.
It would also be interesting to apply the qutrit ZX-calculus to contextuality problem since quantum contextuality exists only in dimensions greater than two \cite{Gleason}.

\section*{Appendix}
\emph{Proof of Lemma \ref{trianglebila}.}\\
\begin{proof}
First we have
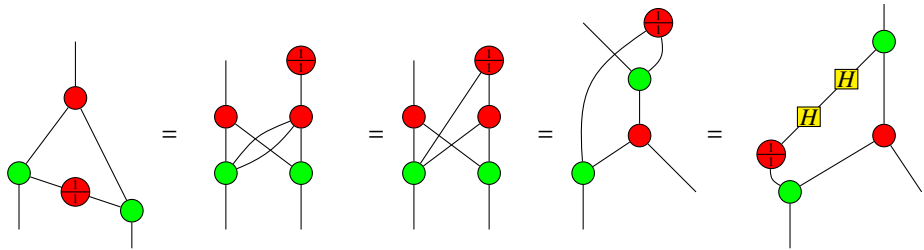
\begin{align*}
\beginpgfgraphicnamed{Qutrits//trianglebilaprf1}
\InputIfFileExists{Qutrits//trianglebilaprf1.tikz}{}{\input{./figures/Qutrits//trianglebilaprf1.tikz}}
\endpgfgraphicnamed
\end{align*}
Here we used commutativity of the green copy, (B2) and Lemma \ref{controlnotslide}.
Since
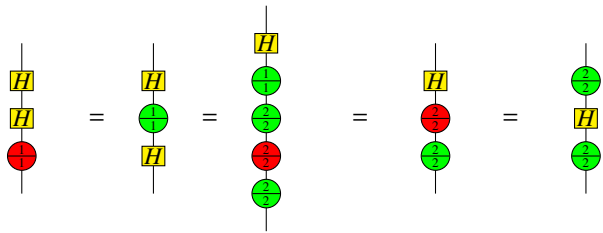
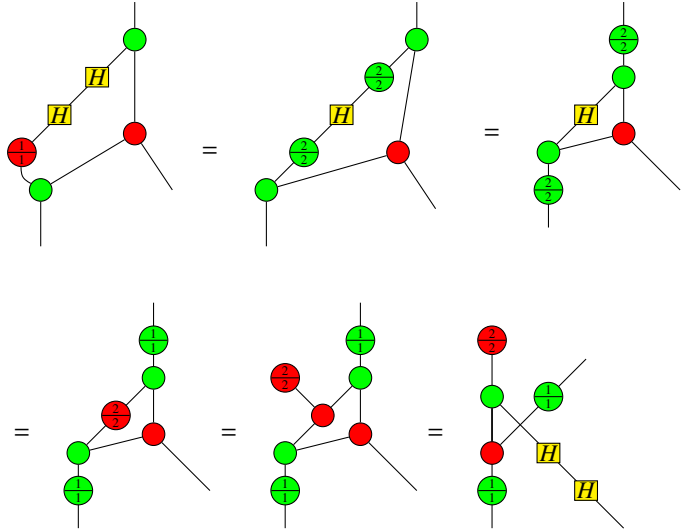
\begin{align*}
\beginpgfgraphicnamed{Qutrits//trianglebilaprf2}
\InputIfFileExists{Qutrits//trianglebilaprf2.tikz}{}{\input{./figures/Qutrits//trianglebilaprf2.tikz}}
\endpgfgraphicnamed
\end{align*}
we have
\begin{align*}
\beginpgfgraphicnamed{Qutrits//trianglebilaprf3}
\InputIfFileExists{Qutrits//trianglebilaprf3.tikz}{}{\input{./figures/Qutrits//trianglebilaprf3.tikz}}
\endpgfgraphicnamed
\end{align*}
Here we used the color-swapped version of lemma \ref{quasibialgebra} for the last equality.
Therefore,
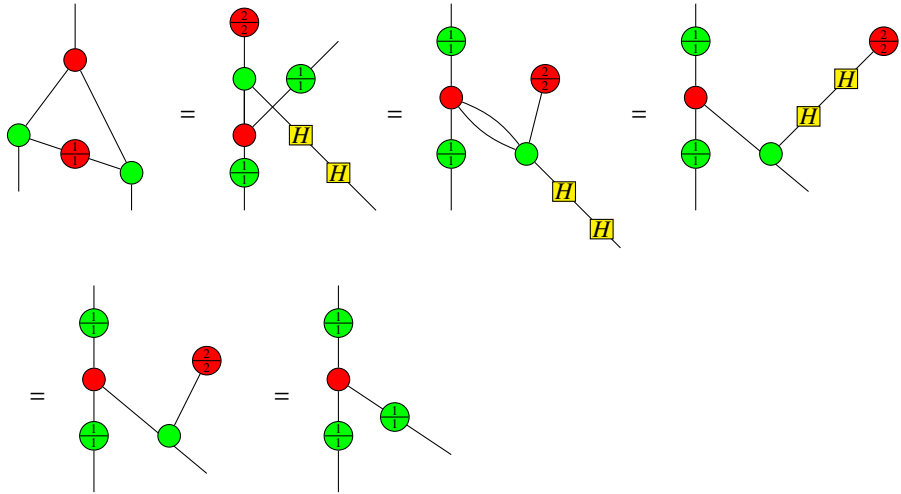
\begin{align*}
\beginpgfgraphicnamed{Qutrits//trianglebilaprf4}
\InputIfFileExists{Qutrits//trianglebilaprf4.tikz}{}{\input{./figures/Qutrits//trianglebilaprf4.tikz}}
\endpgfgraphicnamed
\end{align*}
Here we used lemma \ref{rotationchange} for the last equality.
\end{proof}

\emph{Proof of Lemma \ref{hadimplylocalcomple}.}\\
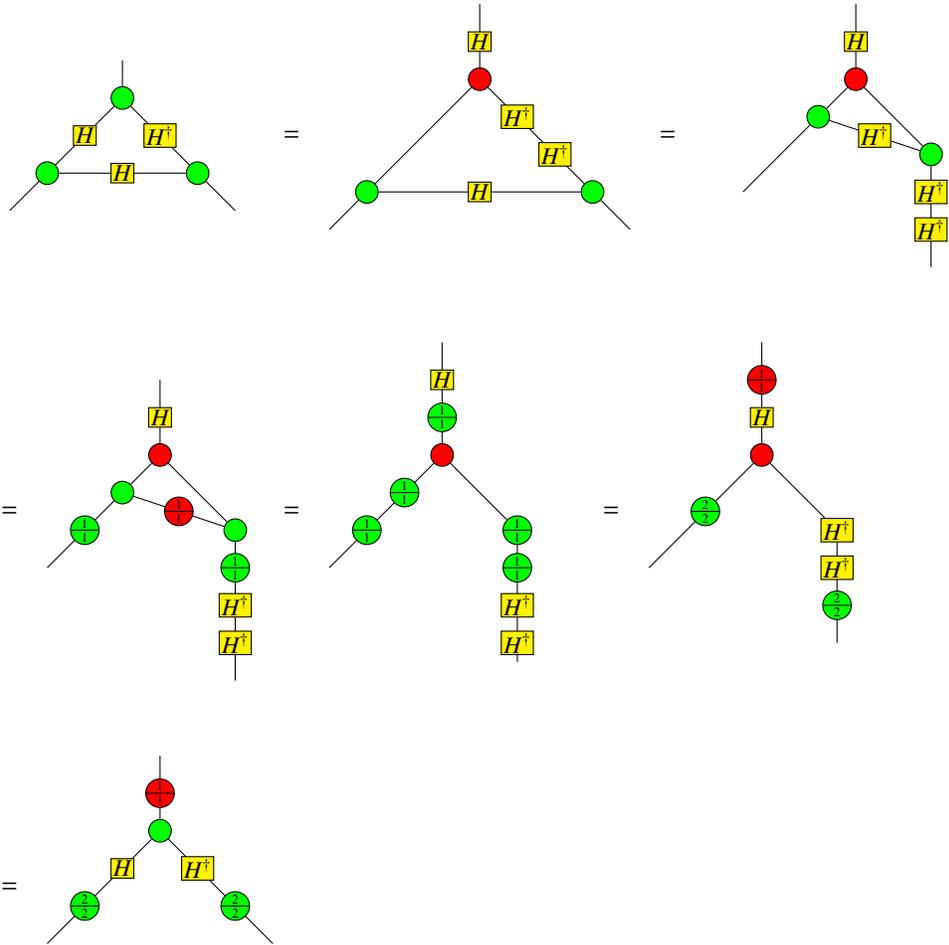
\begin{proof}
\begin{align*}
\beginpgfgraphicnamed{Qutrits//hadimplylocalcompleprf}
\InputIfFileExists{Qutrits//hadimplylocalcompleprf.tikz}{}{\input{./figures/Qutrits//hadimplylocalcompleprf.tikz}}
\endpgfgraphicnamed
\end{align*}
Here we used Euler decomposition and lemma \ref{trianglebila}.
\end{proof}

\emph{Proof of Lemma \ref{hadimplylocalcomple2}.}\\
\begin{proof}
By Lemma \ref{hadimplylocalcomple} we have
\begin{align*}
\beginpgfgraphicnamed{Qutrits//localhhhdprf}
\InputIfFileExists{Qutrits//localhhhdprf.tikz}{}{\input{./figures/Qutrits//localhhhdprf.tikz}}
\endpgfgraphicnamed
\end{align*}

\end{proof}

\emph{Proof of Lemma \ref{hadimplylocalcomplestar}.}\\
\begin{proof}
The proof is by induction on $n$.
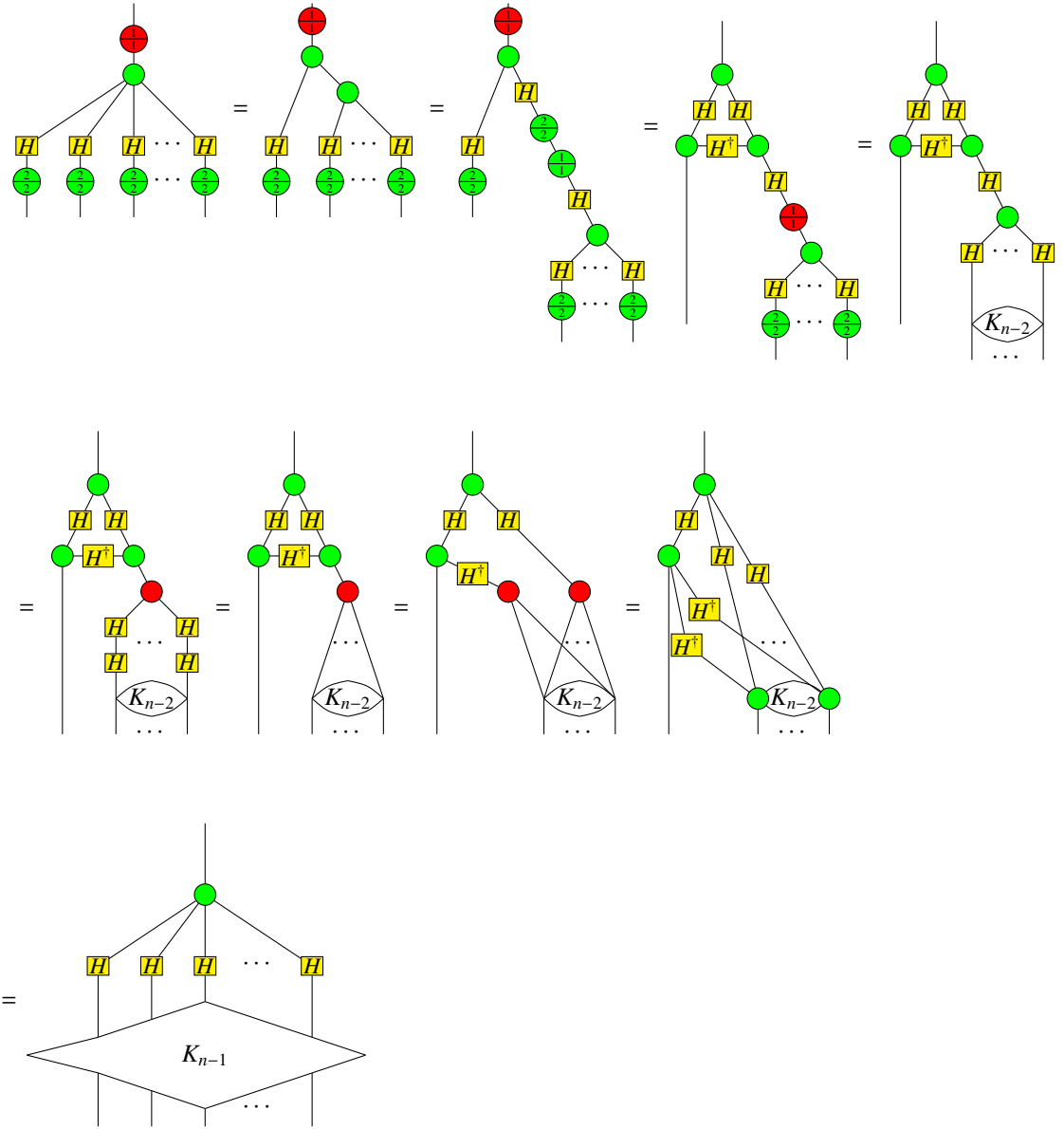
\begin{align*}
\beginpgfgraphicnamed{Qutrits//stargraphprf}
\InputIfFileExists{Qutrits//stargraphprf.tikz}{}{\input{./figures/Qutrits//stargraphprf.tikz}}
\endpgfgraphicnamed
\end{align*}
Here we used generalized bialgebra rules for qutrits, which can be derived similarly to the qubit case as shown in Lemma 3 of \cite{DuncanPerdrix}.
\end{proof}

\end{document}